\documentclass[11pt,a4paper]{article}

\usepackage{amsmath,amsfonts, amssymb}
\begin{document}
\bibliographystyle{abbrv}

\title{On the entropy and log-concavity of compound Poisson measures}
\author{
Oliver Johnson\thanks{Department of Mathematics, University of Bristol, University Walk, Bristol, BS8 1TW, UK.
Email: \texttt{O.Johnson@bristol.ac.uk}}
\and 
Ioannis Kontoyiannis\thanks{Department of Informatics,
Athens University of Economics \& Business,
Patission 76, Athens 10434, Greece.
Email: {\tt yiannis@aueb.gr}}
\and
Mokshay Madiman\thanks{Department of Statistics, Yale University,
24 Hillhouse Avenue, New Haven, CT 06511, USA.
Email: \texttt{mokshay.madiman@yale.edu}}
}
\date{\today}
\maketitle

\newtheorem{theorem}{Theorem}[section]
\newtheorem{lemma}[theorem]{Lemma}
\newtheorem{proposition}[theorem]{Proposition}
\newtheorem{corollary}[theorem]{Corollary}
\newtheorem{conjecture}[theorem]{Conjecture}
\newtheorem{definition}[theorem]{Definition}
\newtheorem{example}[theorem]{Example}
\newtheorem{condition}{Condition}
\newtheorem{main}{Theorem}
\newtheorem{remark}[theorem]{Remark}
\hfuzz25pt

\def \outlineby #1#2#3{\vbox{\hrule\hbox{\vrule\kern #1%
\vbox{\kern #2 #3\kern #2}\kern #1\vrule}\hrule}}%
\def \endbox {\outlineby{4pt}{4pt}{}}%
\newenvironment{proof}
{\noindent{\bf Proof\ }}{{\hfill \endbox
}\par\vskip2\parsep}
\newenvironment{pfof}[2]{\removelastskip\vspace{6pt}\noindent
 {\it Proof  #1.}~\rm#2}{\par\vspace{6pt}}

\newcommand{\vc}[1]{{\mathbf{ #1}}}
\newcommand{\bp}[1]{b_{\vc{#1}}}
\newcommand{\cp}[1]{C_Q b_{\vc{#1}}}
\newcommand{\cbin}[1]{C_Q {\rm Bin}_{#1}}
\newcommand{\cPsa}{C_Q P_{\alpha}^{\#}}
\newcommand{\cP}{C_Q P}
\newcommand{\cPa}{C_Q P_{\alpha}}
\newcommand{\cPx}{C_Q P_{X}}
\newcommand{\cPy}{C_Q P_{Y}}
\newcommand{\cPz}{C_Q P_{X+Y}}
\newcommand{\sco}[1]{r_{1,#1}}
\newcommand{\Pc}{{\mathcal{ P}}}
\newcommand{\wt}[1]{\widetilde{#1}}
\newcommand{\var}{{\rm{Var\;}}}
\newcommand{\cov}{{\rm{Cov\;}}}
\newcommand{\tends}{\rightarrow \infty}
\newcommand{\C}{{\cal C}}
\newcommand{\ep}{{\mathbb {E}}}
\newcommand{\pr}{{\mathbb {P}}}
\newcommand{\re}{{\mathbb {R}}}
\newcommand{\I}{\mathbb {I}}
\newcommand{\Z}{{\mathbb {Z}}}
\newcommand{\tq}{\widetilde{Q}}
\newcommand{\tc}{\widetilde{C}}
\newcommand{\qs}[1]{Q^{\# (#1)}}
\newcommand{\qst}[1]{Q^{* #1}}
\newcommand{\qsh}{Q^{\#}}
\newcommand{\Nat}{{\mathbb {N}}}
\newcommand{\map}[1]{{\bf #1}}
\newcommand{\CPi}{C_Q \Pi}
\newcommand{\pbar}{\vc{\overline{p}}}

\newcommand{\leqa}{\mbox{$ \;\stackrel{(a)}{\leq}\; $}}

\newcommand{\Ch}[2]{\ensuremath{\begin{pmatrix} #1 \\ #2 \end{pmatrix}}} 
\newcommand{\bin}[2]{\binom{#1}{#2}}
\newcommand{\ov}[1]{\overline{#1}}
\newcommand{\bern}[1]{{\rm{Bern}}\left(#1\right)}
\newcommand{\cbern}[2]{{\rm{CBern}}\left(#1,#2\right)}
\newcommand{\geo}[1]{{\rm{Geom}}\left(#1\right)}
\newcommand{\Prob}[1]{\ensuremath{\mathbb{P} \left(#1 \right)}}
\newcommand{\blah}[1]{}
\newcommand{\ch}[1]{{\bf #1}}

\newcommand{\equald}{\mbox{$ \;\stackrel{\cal D}{=}\; $}}

\begin{abstract} 
Motivated, in part, by the desire to develop
an information-theoretic foundation for 
compound Poisson approximation limit theorems
(analogous to the corresponding developments
for the central limit theorem and for simple
Poisson approximation), this work examines
sufficient conditions under which the compound 
Poisson distribution has maximal entropy within 
a natural class of probability measures on the 
nonnegative integers. We show that
the natural analog of the Poisson 
maximum entropy property remains valid if the
measures under consideration are log-concave,
but that it fails in general.
A parallel maximum entropy result is established
for the family of compound binomial measures.
The proofs are largely based on ideas 
related to the semigroup approach introduced in 
recent work by Johnson~\cite{johnson21} for the
Poisson family.
Sufficient conditions are given for compound 
distributions to be log-concave, and specific 
examples are presented illustrating all the above 
results. 

\end{abstract}

\bigskip

{\bf Keywords}

\newpage

\section{Introduction}

A particularly appealing way to state the 
classical central limit theorem is to say that,
if $X_1,X_2,\ldots$ 
are independent and identically distributed,
continuous random variables with zero mean and unit variance, 
then the entropy of their normalized partial sums
$S_n=\frac{1}{\sqrt{n}}\sum_{i=1}^nX_i$ increases 
with $n$ to
the entropy of the standard normal distribution, 
which is maximal among all random variables with
zero mean and unit variance. More precisely, if
$f_n$ denotes the density of $S_n$ and $\phi$ 
the standard normal density, then, as $n\to\infty$,
\begin{equation}
h(f_n)\uparrow h(\phi)
=\sup\{h(f)\;:\;\mbox{densities $f$ with mean 0 and
variance 1}\},
\label{eq:clt}
\end{equation}
where $h(f)=-\int f\log f$ denotes the differential
entropy and log denotes the natural logarithm.
Precise conditions under which 
(\ref{eq:clt}) holds are given in 
\cite{artstein}\cite{tulino}\cite{madiman};
also see
\cite{linnik}\cite{barron}\cite{johnson14}
and the references therein, where numerous
related results are stated, along with
their history.

Part of the appeal of this formalization of the
central limit theorem comes from its analogy
to the second law of thermodynamics: The
``state'' (meaning the distribution)
of the random variables $S_n$ evolves
monotonically, until the {\em maximum entropy}
state, the standard normal distribution, is
reached. Moreover, the introduction of
information-theoretic ideas and techniques
in connection with the entropy has motivated 
numerous related results (and their proofs),
generalizing and strengthening the central
limit theorem in different directions; see
the references mentioned above for details.

The classical Poisson convergence limit theorems, of which
the binomial-to-Poisson is the prototypical 
example, have also been examined under a similar light. 
An analogous program has been recently carried out in
this case 
\cite{shepp}\cite{johnstone}\cite{harremoes}\cite{johnson11}\cite{johnson21}.
The starting point is the identification
of the Poisson distribution as that 
which has maximal entropy within a natural
class of probability measures.
Perhaps the simplest way to state and
prove this is along the following lines; first we
make some simple definitions:
\begin{definition}
For any parameter vector $\vc{p} = (p_1,p_2, \ldots, p_n)$ 
with each $p_i\in[0,1]$,
the sum of independent Bernoulli random variables $B_i\sim\bern{p_i}$,
$$S_n=\sum_{i=1}^n B_i,$$
is called a {\em Bernoulli sum}, and its 
probability mass function is denoted by 
$\bp{p}(x):=\Pr\{S_n=x\}$,
for $x=0,1,\ldots$. Further, for each $\lambda>0$, we define
the following sets of parameter vectors: 
\begin{eqnarray*}
\Pc_n(\lambda) & = & \big\{ \vc{p}\in[0,1]^n
\;:\; p_1+p_2+\cdots+p_n =\lambda 
\big\}
\;\;\;\;
\mbox{and}
\;\;\;\;
\Pc_{\infty}(\lambda) = \bigcup_{n\geq 1} \Pc_n(\lambda).
\end{eqnarray*}
\label{def:bp}
\end{definition}
Shepp and Olkin \cite{shepp} showed 
that, for fixed $n\geq1$,
the Bernoulli sum $\bp{p}$ which has maximal 
entropy among all Bernoulli sums with
mean $\lambda$,
is Bin$(n,\lambda/n)$,
the binomial with parameters $n$ and $\lambda/n$,
\begin{equation}
H(\mbox{Bin}(n,\lambda/n))
=
\max\Big\{ H(\bp{p})\;:\; {\vc{p}\in\Pc_n(\lambda)}\Big\},
\label{eq:maxEntB}
\end{equation}
where $H(P)=-\sum_x P(x)\log P(x)$ denotes the discrete
entropy function. Noting that the binomial
$\mbox{Bin}(n,\lambda/n)$ converges to the Poisson
distribution $\mbox{Po}(\lambda)$ as $n\to\infty$,
and that the classes of Bernoulli sums in (\ref{eq:maxEntB})
are nested, 
$\{\bp{p}:\vc{p}\in\Pc_n(\lambda)\}\subset
\{\bp{p}:\vc{p}\in\Pc_{n+1}(\lambda)\},$ 
Harremo\"es \cite{harremoes} 
noticed that a simple limiting
argument gives the following 
maximum entropy property
for the Poisson distribution:
\begin{equation}
H(\mbox{Po}(\lambda))
=
\sup\Big\{
H(\bp{p})\;:\;
\vc{p}\in\Pc_\infty(\lambda)\Big\}.
\label{eq:maxEntP}
\end{equation}

Partly motivated by the desire to provide
an information-theoretic foundation for 
{\em compound Poisson limit theorems}
and the
more general problem of {\em compound Poisson 
approximation}, as a first step we consider
the problem of generalizing
the maximum entropy
properties (\ref{eq:maxEntB})
and (\ref{eq:maxEntP}) to the case of
{\em compound Poisson} distributions
on ${\mathbb Z}_+$.\footnote{Recall
	that the compound Poisson distributions 
	are the only infinitely divisible distributions 
	on ${\mathbb Z}_+$, and also
	they are (discrete) stable laws
	\cite{steutel2}.
	In the way of motivation we also
	recall Gnedenko and Korolev's
	remark that ``there should be mathematical ...
	probabilistic models of the universal principle
	of non-decrease of uncertainty,''
	and their proposal that we should
	``find conditions under which certain limit 
	laws appearing in limit theorems of probability 
	theory possess extremal entropy properties. Immediate 
	candidates to be subjected to such analysis are, 
	of course, stable laws \ldots'';
	see \cite[pp. 211-215]{gnedenko2}.
	}
We begin with some definitions:

\begin{definition} Let $P$ be an arbitrary distribution
on $\Z_+=\{0,1,\ldots\}$, and $Q$ a distribution on 
$\Nat = \{1, 2, \ldots \}$.
The {\em $Q$-compound distribution $\cP$} is the
distribution of the random sum,
\begin{equation} \label{eq:randsum}
\sum_{j=1}^{Y} X_j,
\end{equation}
where $Y$ has distribution $P$ and the random variables 
$\{X_j\}$ are independent and identically distributed
(i.i.d.) with common distribution $Q$ and 
independent of $Y$.
The distribution $Q$ is called a
{\em compounding distribution},
and the map $P\mapsto C_Q P$ is
the {\em $Q$-compounding operation}.
The $Q$-compound distribution $C_QP$
can be explicitly written as the mixture,
\begin{equation} 
\label{eq:compdis}
 \cP(x) = \sum_{y=0}^{\infty} P(y) \qst{y}(x),
	\;\;\;\;x\geq 0,
\end{equation}
where $Q^{*j}(x)$ is the $j$th convolution power of $Q$ and 
$Q^{*0}$ is the point mass at $x=0$. 
\end{definition}

Above and throughout the paper,
the empty sum $\sum_{j=1}^0(\cdots)$ is taken to be zero;
all random variables considered are supported 
on $\Z_+=\{0,1,\ldots\}$; and all compounding 
distributions $Q$ are supported on $\Nat=\{1,2,\ldots\}$.

\begin{example} Let $Q$ be an arbitrary distribution on $\Nat$.
\begin{enumerate} 
	\item 
For any $0 \leq p \leq 1$, the {\em compound Bernoulli
distribution $\cbern{p}{Q}$} is the distribution
of the product $BX$, where $B\sim\mbox{Bern}(p)$
and $X\sim Q$ are independent.
It has probability mass function
$C_Q P$, where $P$ is the $\bern{p}$ mass function,
so that, $C_Q P(0)=1-p$ and $C_Q P(x)=pQ(x)$ for $x\geq 1$.
	\item 
A {\em compound Bernoulli sum} is a sum of independent 
compound Bernoulli random variables, all with respect 
to the same compounding distribution $Q$: Let 
$X_1,X_2,\ldots,X_n$ be i.i.d.\ with common
distribution $Q$ and $B_1,B_2,\ldots,B_n$
be independent Bern($p_i$). We call,
$$ \sum_{i=1}^n B_iX_i \;\equald\; \sum_{j=1}^{\sum_{i=1}^n B_i} X_j,$$
a {\em compound Bernoulli sum}; in view of~{\em (\ref{eq:randsum})},
its distribution is $\cp{p}$, where
$\vc{p} = (p_1,p_2, \ldots, p_n)$.
	\item
In the special case of a compound Bernoulli sum with
all its parameters $p_i=p$ for a fixed $p\in[0,1]$,
we say that it has a {\em compound binomial distribution},
denoted by $\mbox{\em CBin}(n,p,Q)$.
	\item
Let $\Pi_\lambda(x)=e^{-\lambda}\lambda^x/x!$, $x\geq 0$,
denote the {\em Po}$(\lambda)$ mass function. Then,
for any $\lambda>0$,
the {\em compound Poisson distribution $\mbox{CPo}(\lambda,Q)$}
is the distribution with mass function $\CPi_\lambda$:
\begin{equation} \label{eq:cppmf}
\CPi_{\lambda}(x) =
\sum_{j=0}^{\infty} \Pi_\lambda(j)
Q^{*j}(x) = 
\sum_{j=0}^{\infty} \frac{ e^{-\lambda} \lambda^j}{j!}
Q^{*j}(x),
\;\;\;\;x\geq 0.
\end{equation}
\end{enumerate}
\end{example}

In view of the Shepp-Olkin maximum entropy property (\ref{eq:maxEntB})
for the binomial distribution, a first natural conjecture
might be that the compound binomial has maximum entropy
among all compound Bernoulli sums $\cp{p}$
with a fixed mean; that is, 
\begin{equation}
H(\mbox{CBin}(n,\lambda/n,Q))
=
\max\Big\{ H(C_Q\bp{p})\;:\; {\vc{p}\in\Pc_n(\lambda)}\Big\}.
\label{eq:maxEntBC}
\end{equation}
But, perhaps somewhat surprisingly, as Chi \cite{chi} 
has noted, (\ref{eq:maxEntBC}) fails in general. For example,
taking $Q$ to be the uniform distribution on $\{1,2\}$,
$\vc{p}=(0.00125, 0.00875)$
and $\lambda =p_1+p_2=0.01$, 
direct computation shows that,
\begin{equation}
H(\mbox{CBin}(2,\lambda/2,Q))
<0.090798
<0.090804
< H(C_Q\bp{p}).
\label{eq:chi}
\end{equation}

As the Shepp-Olkin result (\ref{eq:maxEntB})
was only seen as an intermediate step in proving
the maximum entropy property of the Poisson 
distribution (\ref{eq:maxEntP}), we may still
hope that the corresponding result remains
true for compound Poisson measures,
namely that,
\begin{equation}
H(\mbox{CPo}(\lambda,Q))
=
\sup\Big\{
H(C_Q\bp{p})\;:\;
\vc{p}\in\Pc_\infty(\lambda)\Big\}.
\label{eq:maxEntPC}
\end{equation}
Again, (\ref{eq:maxEntPC}) fails in general. 
For example, taking the same 
$Q,\lambda$ and $\vc{p}$ as above,
yields,
$$
H(\mbox{CPo}(\lambda,Q))
<0.090765
<0.090804
< H(C_Q\bp{p}).$$

The main purpose of the present work
is to show that, despite these negative
results, it is possible to provide
natural, broad sufficient conditions,
under which the compound binomial and 
compound Poisson distributions can be
shown to have maximal entropy in an
appropriate class of measures.
Our first result, Theorem~\ref{thm:mainber}
below, states that (\ref{eq:maxEntBC})
{\em does} hold, under certain
conditions on $Q$ and CBin($n,\lambda,Q$):

\begin{theorem} \label{thm:mainber}
If the distribution $Q$ on $\Nat$ 
and the compound binomial distribution 
$\mbox{\em CBin}(n,\lambda/n,Q)$
are both log-concave, 
then,
$$H(\mbox{\em CBin}(n,\lambda/n,Q))
=\max\Big\{ H(C_Q\bp{p})\;:\; {\vc{p}\in\Pc_n(\lambda)}\Big\},$$
as long as the tail of $Q$ satisfies
either one of the following properties:
$(a)$~$Q$ has finite support; or 
$(b)$~$Q$ has tails heavy enough so that,
for some $\rho,\beta>0$ and $N_0\geq 1$, 
we have, $Q(x)\geq \rho^{x^\beta}$,
for all $x\geq N_0$.
\end{theorem}

The proof of the theorem is given in Section~\ref{sec:compbin}.
As can be seen there,
conditions $(a)$ and $(b)$ are introduced 
purely for technical reasons, and can probably
be significantly relaxed. The notion of log-concavity,
on the other hand, is central in the development
of the ideas in this work. [In a different setting,
log-concavity also appears as a natural condition
for a different maximum entropy problem considered
by Cover and Zhang \cite{cover2}.] Recall that 
the distribution $P$ of a random variable $X$ on $\Z_+$
is {\em log-concave} if its support is a (possibly infinite)
interval of successive integers in $\Z_+$, and,
\begin{equation} \label{eq:lcdef}
P(x)^2 \geq P(x+1) P(x-1),\;\;\;\; \mbox{for all $x \geq 1$.}
\end{equation}
We also recall that most of the 
commonly used distributions appearing 
in applications (e.g.,
the Poisson, binomial, geometric, negative binomial, hypergeometric
logarithmic series, or Polya-Eggenberger distribution)
are log-concave.

Another key property is that of
ultra log-concavity;
cf.\ \cite{pemantle}. The distribution $P$ of a random variable 
$X$ is {\em ultra log-concave} if $P(x)/\Pi_{\lambda}(x)$ is 
log-concave, that is, if,
\begin{equation} \label{eq:ulcdef}
x P(x)^2 \geq (x+1) P(x+1) P(x-1),\;\;\;\; \mbox{for all $x \geq 1$.}
\end{equation}
Note that the Poisson distribution as well as all Bernoulli sums 
are ultra log-concave.

Johnson \cite{johnson21} recently
proved the following maximum entropy 
property for the Poisson distribution,
generalizing (\ref{eq:maxEntP}):
\begin{equation}
H(\mbox{Po}(\lambda))
=
\max\Big\{
H(P)\;:\;
\mbox{ultra log-concave $P$ with mean $\lambda$}
\Big\}.
\label{eq:maxEntPJ}
\end{equation}
Our next result (proved in Section~\ref{sec:comppoi})
states that,
as long as $Q$ and the compound Poisson measure
$\mbox{CPo}(\lambda,Q)$ are log-concave, 
the same maximum entropy statement as in 
(\ref{eq:maxEntPJ}) remains valid 
in the compound Poisson case:

\begin{theorem} \label{thm:mainpoi}
If the distribution $Q$ on $\Nat$ and 
and the compound Poisson distribution
$\mbox{\em CPo}(\lambda,Q)$ are both log-concave, 
then,
$$ 
H(\mbox{\em CPo}(\lambda,Q))
=\max\Big\{
H(C_Q P) \;:\; \mbox{ultra log-concave $P$ with mean $\lambda$}
\Big\}.$$
\end{theorem}

In Section~\ref{sec:lccond} we give conditions under which 
the compound Poisson and compound Bernoulli distributions 
are log-concave. In particular, the results there
imply the following explicit maximum entropy statements.

\begin{example}
\begin{enumerate}
\item
Let $Q$ be an arbitrary log-concave distribution
on ${\mathbb N}$. Then Lemma~\ref{lem:lc} 
combined with Theorem~\ref{thm:mainber} implies that
the maximum entropy property of the
compound binomial distribution in
equation~{\em (\ref{eq:maxEntBC})} holds,
for all $\lambda$ large enough. That is,
the compound binomial 
{\em CBin($n,\lambda/n,Q$)} has maximal entropy
among all compound Bernoulli sums $C_Q\bp{p}$
with $p_1+p_2+\cdots+p_n=\lambda$, as long
as $\lambda \geq \frac{nQ(2)}{Q(1)^2+Q(2)}$.
\item
Suppose $Q$ is supported on $\{1,2\}$, 
with probabilities $Q(1)=q,Q(2)=1-q$,
and consider the class of all Bernoulli sums
$\bp{p}$ with mean $p_1+p_2+\cdots+p_n=\lambda$.
Theorem~\ref{thm:q2pt} combined with
Theorem~\ref{thm:mainpoi} implies that the
compound Poisson maximum entropy property
{\em (\ref{eq:maxEntPC})} holds in this case,
as long as $\lambda$
is large enough. In other words,
the distribution {\em CPo($\lambda,Q$)} has
maximal entropy 
among all compound Bernoulli sums $C_Q\bp{p}$
with $p_1+p_2+\cdots+p_n=\lambda\geq
\frac{2(1-q)}{q^2}$.
\item
Suppose $Q$ is geometric with parameter $\alpha\in(0,1)$,
i.e., $Q(x)=\alpha(1-\alpha)^{x-1}$ for all $x\geq 1$,
and again consider the class of  a Bernoulli sums $\bp{p}$
with mean $\lambda$.
Then Theorem~\ref{thm:qgeom} combined with
Theorem~\ref{thm:mainpoi} implies that 
{\em (\ref{eq:maxEntPC})} holds for all large $\lambda$:
The compound Poisson 
distribution {\em CPo($\lambda,Q$)} has
maximal entropy 
among all compound Bernoulli sums $C_Q\bp{p}$
with $p_1+p_2+\cdots+p_n=\lambda\geq
\frac{2(1-\alpha)}{\alpha}$.
\end{enumerate}
\end{example}

Clearly, it remains an open question to give {\em necessary} 
and sufficient conditions on $\lambda$ and $Q$ for the compound 
Poisson and compound binomial distributions to have maximal 
entropy within an appropriately defined class, or even
for the compound Poisson distribution to be log-concave.
Section~4 ends with a conjecture, together with some
supporting evidence, stating that CPo$(\lambda,Q)$ is
log-concave when $Q$ is log-concave and $\lambda Q(1)^2\geq 2Q(2)$.

\newpage

\section{Maximum Entropy Property of the Compound Poisson Distribution} 
\label{sec:comppoi}

Here we show that, if $Q$ and the
compound Poisson distribution 
$\mbox{CPo}(\lambda,Q)=C_Q\Pi_\lambda$ 
are both log-concave, then 
$\mbox{CPo}(\lambda,Q)$
has maximum entropy among all 
distributions of the form $\cP$, when $P$ has mean 
$\lambda$ and is ultra log-concave.
Our approach is an extension of the 
`semigroup' arguments of \cite{johnson21}.

We begin by recording some basic properties
of log-concave and ultra log-concave distributions:
\begin{itemize}
\item[$(i)$]
If $P$ is ultra log-concave, then
from the definitions it is immediate
that $P$ is log-concave.
\item[$(ii)$]
If $Q$ is log-concave, then it has finite moments
of all orders; see \cite[Theorem~7]{keilson}.
\item[$(iii)$]
If $X$ is a random variable
with ultra log-concave distribution $P$, then (by~$(i)$ 
and~$(ii)$) it has finite moments of all orders.
Moreover, considering the covariance between the decreasing 
function $P(x+1) (x+1)/P(x)$ and the increasing function
$x(x-1) \cdots (x-n)$, shows that the falling 
factorial moments 
of $P$ satisfy, 
$$E[(X)_n]:=E[X(X-1) \cdots (X-n+1)] \leq (E(X))^n;$$ 
see \cite{johnson21} and \cite{johnsonc2}
for details. 
\item[$(iv)$] 
The Poisson distribution and all Bernoulli
sums are ultra log-concave.
\end{itemize}

Recall the following definition 
from \cite{johnson21}:

\begin{definition} 
\label{def:stmap} 
%
Given $\alpha\in[0,1]$ and a random variable $X\sim P$ 
on $\Z_+$ with mean $\lambda\geq 0$,
let $U_\alpha P$ denote the 
distribution of the random variable,
$$\sum_{i=1}^X B_i+Z_{\lambda(1-\alpha)},$$
where the $B_i$ are i.i.d.\ $\bern{\alpha}$,
$Z_{\lambda(1-\alpha)}$ has distribution $\mbox{\em Po}(\lambda(1-\alpha))$,
and all random variables are independent
of each other and of $X$.
\end{definition}

Note that, if $X\sim P$ has mean $\lambda$,
then $U_\alpha P$ has the same mean. Also,
recall the following useful relation that
was established in 
Proposition~3.6 of \cite{johnson21}: For all $y\geq 0$,
\begin{equation} \label{eq:newheat}
\frac{\partial }{\partial \alpha} U_{\alpha}P(y)  = \frac{1}{\alpha} \left(
\lambda
(U_{\alpha}P(y) - U_{\alpha}P(y-1)
- ((y+1) U_{\alpha}P(y+1) - y U_{\alpha}P(y)) \right).
\end{equation}
%
Next we define another transformation of 
probability distributions $P$ on ${\mathbb Z}_+$: 

\begin{definition} \label{def:clustermap}
Given $\alpha\in[0,1]$, a distribution $P$ on $\Z_+$
and a compounding distribution $Q$ on $\Nat$, 
let $U^Q_{\alpha}P$ denote the
distribution $C_Q U_\alpha P$:
$$U^Q_{\alpha} P(x):=C_QU_\alpha P(x)
= \sum_{y=0}^{\infty} U_\alpha P(y) Q^{*y}(x),
\;\;\;\;x\geq 0.$$
\end{definition}

An important observation that will be at the
heart of the proof of Theorem~\ref{thm:mainpoi}
below is that, for $\alpha=0$, $U_0^QP$
is simply the compound Poisson
measure CP$(\lambda,Q)$, while for $\alpha=1$,
$U_1^QP=C_QP$. The following lemma, proved in the
appendix, gives a rough bound on the third
moment of $U_\alpha^QP$:

\begin{lemma} \label{lem:moments}
Suppose $P$ is an ultra log-concave 
distribution with mean $\lambda>0$ 
on ${\mathbb Z}_+$, and
let $Q$ be a log-concave compounding 
distribution on ${\mathbb N}$.
For each $\alpha\in[0,1]$,
let $W_\alpha,V_\alpha$ be random variables
with distributions $U_\alpha^QP=C_Q U_\alpha P$
and $C_Q (U_\alpha P)^{\#}$, respectively,
where, for any distribution $R$ with mean $\nu$, 
we write $R^{\#}(y) = R(y+1)(y+1)/\nu$ 
for its {\em size-biased} version.
Then the third moments 
$E(W_\alpha^3)$
and $E(V_\alpha^3)$ are both bounded above by,
$$\lambda q_3 +3\lambda^2q_1q_2+\lambda^3q_1^3,$$
where $q_1,q_2,q_3$ denote the first, second and third
moments of $Q$, respectively.
\end{lemma}

In \cite{johnson21}, the characterization
of the Poisson as a maximum entropy 
distribution was proved through
the decrease of its score function. In
an analogous way, 
following \cite{johnson22}, 
we define the score function 
of a $Q$-compound random variable as follows. 

\begin{definition} \label{def:score}
Given a distribution $P$ on $\Z_+$ with mean $\lambda$, 
the corresponding $Q$-compound distribution
$\cP$ has score function defined by:
\begin{equation} \label{eq:score}
 \sco{\cP}(x) = \frac{ \sum_{y=0}^{\infty} (y+1) P(y+1) \qst{y}(x)}{\lambda
\sum_{y=0}^{\infty} P(y) \qst{y}(x) } - 1 = 
\frac{ \sum_{y=0}^{\infty} (y+1) P(y+1) \qst{y}(x)}{\lambda \cP(x)  } - 1.
\end{equation}
\end{definition}

Notice that the mean of 
of $\sco{\cP}$ with respect to $\cP$ is zero,
and that if $P\sim\mbox{Po}(\lambda)$
then $\sco{\cP}(x) \equiv 0$. Further, 
when $Q$ is the point mass at 1
this score function 
reduces to the ``scaled score function'' introduced in \cite{johnson11}.
But, unlike the scaled score function and 
the alternative score function $r_{2,\cP}$ given in
\cite{johnson22}, this score function is not only a function 
of the compound distribution $\cP$, but also explicitly depends
on $P$. A projection identity and other properties of 
$\sco{\cP}$ are proved in \cite{johnson22}.

Next we show that, if $Q$ is log-concave and $P$
is ultra log-concave, then the score function 
$\sco{\cP}(x)$ is decreasing in $x$.

\begin{lemma} \label{lem:decsc}
If $P$ is ultra log-concave and the compounding 
distribution $Q$ is log-concave,
then the score function $\sco{\cP}(x)$ of 
$\cP$ is decreasing in $x$. 
\end{lemma}
\begin{proof} First we recall Theorem~2.1 of Keilson 
and Sumita \cite{keilson2}, 
which implies that,
if $Q$ is log-concave, then for any $m \geq n$, and for any $x$:
\begin{equation} \label{eq:tech}
 \qst{m}(x+1) \qst{n}(x) - \qst{m}(x) \qst{n}(x+1) \geq 0. \end{equation}
[This can be proved by
considering $\qst{m}$ as the convolution of $\qst{n}$ and $\qst{(m-n)}$, 
and writing
\begin{eqnarray*}
\lefteqn{ \qst{m}(x+1) \qst{n}(x) - \qst{m}(x) \qst{n}(x+1)  } \\
& = & \sum_l \qst{(m-n)}(l) \bigg( \qst{n}(x+1-l) \qst{n}(x) -
\qst{n}(x-l) \qst{n}(x+1) \bigg).
\end{eqnarray*}
Since $Q$ is log-concave, then 
so is $\qst{n}$, 
cf.\ \cite{karlin3},
so the ratio $\qst{n}(x+1)/\qst{n}(x)$ is decreasing in $x$, and  
(\ref{eq:tech}) follows.]

By definition, $\sco{\cP}(x) \geq \sco{\cP}(x+1)$ if and only if,
\begin{eqnarray}
0 & \leq & \left( \sum_y (y+1) P(y+1) \qst{y}(x) \right) \left(
\sum_z P(z) \qst{z}(x+1) \right) \nonumber \\
& &  - \left( \sum_y (y+1) P(y+1) \qst{y}(x+1) \right) 
\left( \sum_z P(z) \qst{z}(x) \right) \nonumber \\
& = & \sum_{y,z} (y+1) P(y+1) P(z) \left[ \qst{y}(x) \qst{z}(x+1)
- \qst{y}(x+1) \qst{z}(x) \right]. \label{eq:doublesum}
\end{eqnarray}
Noting that for $y=z$ the term in square brackets in the
double sum becomes zero, and swapping the values of $y$ and
$z$ in the range $y>z$,
the double sum in
(\ref{eq:doublesum}) becomes,
$$ \sum_{y < z} [(y+1) P(y+1) P(z) - (z+1) P(z+1) P(y)]
\left[ \qst{y}(x) \qst{z}(x+1)
- \qst{y}(x+1) \qst{z}(x) \right].$$
By the ultra log-concavity of $P$, the first square 
bracket is positive for $y \leq z$,
and by equation~(\ref{eq:tech}) the second square bracket is 
also positive for $y \leq z$.
\end{proof}

We remark that, under the same assumptions, and using a very similar 
argument, an analogous result holds for the score function $r_{2,\cP}$
recently introduced in \cite{johnson22}.

Combining Lemmas~\ref{lem:decsc} and~\ref{lem:moments}
with equation~(\ref{eq:newheat}) 
we deduce the following result,
which is the main technical step
in the proof of Theorem~\ref{thm:mainpoi} below.

\begin{proposition} \label{prop:deriv}
Let $P$ be an ultra log-concave distribution on ${\mathbb Z}_+$
with mean $\lambda>0$, and assume that 
$Q$ and $\mbox{\em CPo}(\lambda,Q)$ are
both log-concave. Let $W_\alpha$ be a 
random variable with  distribution $U_\alpha^QP$, 
and define, for all $\alpha\in[0,1],$
the function,
$$E(\alpha):=E[-\log \CPi_{\lambda}(W_{\alpha})].$$
Then $E(\alpha)$ is continuous for all $\alpha\in[0,1]$,
it is differentiable for $\alpha\in(0,1)$, and,
moreover, $E'(\alpha)\leq 0$ for $\alpha\in(0,1)$.
In particular, $E(0)\geq E(1)$.
\end{proposition}

\begin{proof} 
Recall that, 
$$ 
U^Q_{\alpha}P(x) = 
C_Q U_\alpha P(x)
= 
\sum_{y=0}^{\infty} U_{\alpha}P(y) \qst{y}(x)
=\sum_{y=0}^{x} U_{\alpha}P(y) \qst{y}(x),
$$
where the last sum is restricted to the
range $0\leq y\leq x$, because
$Q$ is supported on $\Nat$.
Therefore, since $U_\alpha P(x)$ is continuous
in $\alpha$ \cite{johnson21},
so is $U_\alpha^Q P(x)$,
and to show that $E(\alpha)$ is continuous
it suffices to show that the series,
\begin{eqnarray}
E(\alpha)
:=E[-\log \CPi_{\lambda}(W_{\alpha})]
=-\sum_{x=0}^\infty U_\alpha^QP(x)\log\CPi_\lambda(x),
\label{eq:series}
\end{eqnarray}
converges uniformly. To that end,
first observe that log-concavity of $C_Q\Pi_\lambda$
implies that $Q(1)$ is nonzero. [Otherwise,
if $i>1$ be the smallest integer $i$ such that $Q(i)\neq 0$, 
then $C_Q\Pi_\lambda(i+1)=0$, but 
$C_Q\Pi_\lambda(i)$ and
$C_Q\Pi_\lambda(2i)$ are both strictly positive,
contradicting the log-concavity of 
$C_Q\Pi_\lambda$.]
Since $Q(1)$ is nonzero, we can bound the
compound Poisson probabilities as, 
$$1 \geq \CPi_{\lambda}(x) = \sum_{y} [e^{-\lambda} \lambda^y/y!]\qst{y}(x)
\geq e^{-\lambda} [\lambda^x/x!] Q(1)^x, 
\;\;\;\;\mbox{for all}\;x\geq 1,$$
so that the summands in (\ref{eq:series})
can be bounded,
\begin{equation} \label{eq:boundlog}
0 \leq - \log \CPi_{\lambda}(x) \leq \lambda + \log x! - x \log( \lambda Q(1))
\leq Cx^2,
\;\;\;\;x\geq 1,\end{equation} 
for a constant $C>0$ that depends only on $\lambda$ and $Q(1)$.
Therefore, for any $N\geq 1$, the tail of the series (\ref{eq:series})
can be bounded,
$$
0\leq -\sum_{x=N}^\infty U_\alpha^QP(x)\log\CPi_\lambda(x)
\leq C E[W^2_\alpha{\mathbb I}_{\{W_\alpha\geq N\}}]
\leq \frac{C}{N}E[W_\alpha^3],$$
and, in view of Lemma~\ref{lem:moments}, 
it converges uniformly.

Therefore, $E(\alpha)$ is continuous in $\alpha$, 
and, in particular, convergent for all $\alpha\in[0,1]$.
To prove that it is differentiable at each $\alpha\in(0,1)$
we need to establish that: (i)~the summands in (\ref{eq:series})
are continuously differentiable in $\alpha$ for each $x$; 
and (ii)~the series
of derivatives converges uniformly. 

Since, as noted above, $U_\alpha^Q P(x)$ is defined
by a finite sum, we can differentiate with respect 
to $\alpha$ under the sum, to obtain,
\begin{eqnarray}
\frac{\partial}{\partial \alpha} 
U^Q_{\alpha} P(x)
=
\frac{\partial}{\partial \alpha} 
C_Q U_\alpha P(x)
= \sum_{y=0}^{x} \frac{\partial}{\partial \alpha} 
U_{\alpha}P(y) \qst{y}(x).
\label{eq:finite}
\end{eqnarray}
And since $U_\alpha P$ is continuously 
differentiable in $\alpha\in(0,1)$
for each $x$ (cf.\ \cite[Proposition~3.6]{johnson21}
or equation (\ref{eq:newheat}) above),
so are the summands in (\ref{eq:series}),
establishing~(i); in fact, they are
infinitely differentiable, which can be seen 
by repeated applications of (\ref{eq:newheat}). 
To show that the
series of derivatives converges uniformly,
let $\alpha$ be restricted in an arbitrary
open interval $(\epsilon,1)$ for some $\epsilon>0$.
The relation (\ref{eq:newheat})
combined with (\ref{eq:finite}) yields,
for any $x$,
\begin{eqnarray}
\lefteqn{\frac{\partial}{\partial \alpha} U_\alpha^Q P (x)} \nonumber \\
& = & \sum_{y=0}^{x} 
\biggl( \lambda
(U_{\alpha}P(y) - U_{\alpha}P(y-1)
- ((y+1) U_{\alpha}P(y+1) - y U_{\alpha}P(y)) \biggr)
\qst{y}(x) \nonumber \\
& = & -\frac{1}{\alpha}  \sum_{y=0}^{x} 
\left( (y+1) U_{\alpha}P(y+1) - \lambda U_{\alpha}P(y) \right)
 (\qst{y}(x) - \qst{y+1}(x)) \nonumber \\
& = & - \frac{1}{\alpha} \sum_{y=0}^{x} 
\left( (y+1) U_{\alpha}P(y+1) - \lambda U_{\alpha}P(y) \right)
\qst{y}(x) \nonumber \\
& &  + \sum_{v=0}^{x} Q(v) \frac{1}{\alpha} \sum_{y=0}^{x} 
\left( (y+1) U_{\alpha}P(y+1) - \lambda U_{\alpha}P(y) \right)
\qst{y}(x-v) \nonumber \\
& = & - \frac{\lambda}{\alpha} U_\alpha^Q P(x) \left( 
\frac{ \sum_{y=0}^{x} (y+1) U_{\alpha}P(y+1) \qst{y}(x)}{\lambda
U_\alpha^Q P(x)  } - 1 \right) \nonumber \\
&  &  + \frac{\lambda}{\alpha}
\sum_{v=0}^{x} Q(v) U_\alpha^Q P(x-v) \left( 
\frac{ \sum_{y=0}^{x} (y+1) U_{\alpha}P(y+1) \qst{y}(x-v)}{\lambda
U_\alpha^Q P(x-v)  } - 1 \right) \nonumber \\
& = & - \frac{\lambda}{\alpha} \left( U_\alpha^Q P(x) \sco{U_\alpha^Q P}(x)
- \sum_{v=0}^{x} Q(v) U_\alpha^Q P(x-v) \sco{U_\alpha^Q P}(x-v) \right) .
\label{eq:derivative} \end{eqnarray}
Also,
for any $x$, by definition,
$$|U_\alpha^Q P(x) \sco{U_\alpha^Q P}(x)| 
\leq 
C_Q(U_\alpha P)^{\#}(x) +U_\alpha^QP(x), 
$$ 
where, for any distribution $P$, we write
$P^{\#}(y) = P(y+1)(y+1)/\lambda$ for its size-biased version.
Hence for any $N\geq 1$, equations
(\ref{eq:derivative}) and (\ref{eq:boundlog}) yield the bound,
\begin{eqnarray*}
\lefteqn{
\left| \sum_{x=N}^{\infty} \frac{\partial}{\partial \alpha} 
	U_\alpha^Q P(x) \log \CPi_{\lambda}(x) \right| } \\
& \leq & 
	\sum_{x=N}^{\infty} 
	\frac{C\lambda x^2}{\alpha}
	\Big\{ 
	C_Q(U_\alpha P)^{\#}(x) +U_\alpha^QP(x)
	+ \sum_{v=0}^{x} Q(v)
	[
	C_Q(U_\alpha P)^{\#}(x-v) +U_\alpha^QP(x-v)
	] \Big\}\\
& = & 
	\frac{2C}{\alpha}
	E\Big[
	\Big(
	V_\alpha^2+W_\alpha^2+X^2
	+XV_\alpha+XW_\alpha
	\Big)
	{\mathbb I}_{\{V_\alpha\geq N,\;W_\alpha\geq N,\;X\geq N\}}
	\Big]\\
&\leq&
	\frac{C'}{\alpha}
	\Big\{
	E[V_\alpha^2 {\mathbb I}_{\{V_\alpha\geq N\}}]
	+E[W_\alpha^2 {\mathbb I}_{\{W_\alpha\geq N\}}]
	+E[X^2 {\mathbb I}_{\{X\geq N\}}]
	\Big\}\\
&\leq&
	\frac{C'}{N\alpha}
	\Big\{
	E[V_\alpha^3]
	+E[W_\alpha^3]
	+E[X^3]
	\Big\},
\end{eqnarray*}
where $C,C'>0$ are appropriate finite constants, 
and the random variables 
$V_\alpha\sim C_Q(U_\alpha P)^{\#}$,
$W_\alpha\sim U^Q_\alpha P$ and $X\sim Q$ are independent.
Lemma~\ref{lem:moments} implies that this bound
converges to zero uniformly in $\alpha\in(\epsilon,1)$, as $N\to\infty$.
Since $\epsilon>0$ was arbitrary,
this establishes that 
$E(\alpha)$ is differentiable for all $\alpha\in(0,1)$
and, in fact, that we can differentiate the 
series (\ref{eq:series})
term-by-term, to obtain,
\begin{eqnarray}
\lefteqn{ 
E'(\alpha) 
\;=\;
	- \sum_{x=0}^{\infty} \frac{\partial}{\partial \alpha} 
	U_\alpha^Q P(x) \log \CPi_{\lambda}(x)
	} 
	\label{eq:step1}  \\
& = & 
	\frac{\lambda}{\alpha} \sum_{x=0}^{\infty}
	\left( U_\alpha^Q P(x) \sco{U_\alpha^Q P}(x)
	- \sum_{v=0}^{x} Q(v) U_\alpha^Q P(x-v) \sco{U_\alpha^Q P}(x-v) \right)
 	\log \CPi_{\lambda}(x) \nonumber \\
& = & 
	\frac{\lambda}{\alpha} \sum_{x=0}^{\infty} 
	U_\alpha^Q P(x) \sco{U_\alpha^Q P}(x)
	\left( \log \CPi_{\lambda}(x) - \sum_{v=0}^{\infty} Q(v) 
	\log \CPi_{\lambda}(x+v) 
	\right),
	\nonumber
\end{eqnarray}
where the second equality follows from using
(\ref{eq:derivative}) above, and the rearrangement 
leading to the third equality follows by interchanging
the order of (second) double summation and replacing $x$ 
by $x+v$.

Now we note that, exactly as in \cite{johnson21}, 
the last series above is the covariance between
the (zero-mean) function $\sco{U_\alpha^Q P}(x)$ 
and the function $\left( \log \CPi_{\lambda}(x) 
- \sum_v Q(v) \log \CPi_{\lambda}(x+v) \right)$,
under the measure $U_\alpha^Q P$.
Since $P$ is ultra log-concave, so is $U_\alpha P$
\cite{johnson21}, hence the score function
$\sco{U_\alpha^Q P}(x)$ is decreasing in $x$, 
by Lemma~\ref{lem:decsc}. Also, the 
log-concavity of $\CPi_{\lambda}$ implies that the 
second function is increasing, and
Chebyshev's rearrangement lemma
implies that the covariance is 
less than or equal to zero, proving
that $E'(\alpha)\leq 0$, as claimed.

Finally, the fact that $E(0)\geq E(1)$
is an immediate consequence of the
continuity of $E(\alpha)$ on $[0,1]$
and the fact that 
$E'(\alpha)\leq 0$ for all $\alpha\in(0,1)$.
\end{proof}

Notice that, for the above proof to work, it is not necessary that
$\CPi_{\lambda}$ be log-concave; the weaker property
that $\left( \log \CPi_{\lambda}(x) 
- \sum_v Q(v) \log \CPi_{\lambda}(x+v) \right)$ be increasing is enough.

\begin{proof}{\bf of Theorem~\ref{thm:mainpoi}}
%
As in Proposition~\ref{prop:deriv},
let $W_\alpha\sim U^Q_\alpha P=C_Q U_\alpha P$,
and let $D(P\|Q)$ denote the relative entropy
between $P$ and $Q$, 
$$D(P\|Q):=\sum_{x\geq 0}P(x)\log\frac{P(x)}{Q(x)}.$$
Then, noting that $W_0\sim \CPi_\lambda$ and $W_1\sim C_Q P$,
we have,
\begin{eqnarray*}
H(C_Q P)
&\leq&
	H(C_Q P)+D(C_Q P\|\CPi_\lambda)\\
&=&
	-E[\log \CPi_\lambda(W_1)]\\
&\leq&
	-E[\log \CPi_\lambda(W_0)]\\
&=&
	H(\CPi_\lambda),
\end{eqnarray*}
where the first inequality is simply 
the nonnegativity of relative entropy,
and the second inequality is exactly
the statement that $E(1)\leq E(0)$,
proved in Proposition~\ref{prop:deriv}.
\end{proof}

\newpage

\section{Maximum Entropy Property of the Compound Binomial Distribution} 
\label{sec:compbin}

Here we prove the maximum entropy result for compound 
binomial random variables, Theorem~\ref{thm:mainber}. 
The proof, to some extent, parallels some
of the arguments in \cite{harremoes}\cite{mateev}\cite{shepp},
which rely on differentiating the compound-sum probabilities
$\bp{p}(x)$ 
for a given parameter vector $\vc{p}=(p_1,p_2,\ldots,p_n)$
(recall Definition~\ref{def:bp} in the Introduction),
with respect to an individual $p_i$.
Using the representation,
\begin{equation} \label{eq:master}
\cp{p}(y) = 
\sum_{x=0}^n \bp{p}(x) Q^{*x}(y),
\;\;\;\;y\geq 0,
\end{equation}
differentiating $\cp{p}(x)$
reduces to differentiating
$\bp{p}(x)$,
and leads to an expression
equivalent to that derived
earlier 
in (\ref{eq:derivative})
for the derivative of $C_Q U_\alpha P$
with respect to $\alpha$.
\begin{lemma} \label{lem:partials}
Given a parameter vector $\vc{p}=(p_1,p_2,\ldots,p_n)$,
with $n\geq 2$ and
each $0 \leq p_i \leq 1$, 
let, 
$$ \vc{p_t} = \left( \frac{p_1 + p_2}{2} + t, \frac{p_1 + p_2}{2} - t, p_3, \ldots, p_n \right),$$
for $t \in [-(p_1+p_2)/2, (p_1 + p_2)/2]$. Then,
\begin{equation} \label{eq:maindiff}
\frac{\partial}{\partial t} \cp{p_t}(x) =  (- 2t) 
\sum_{y=0}^n \bp{\wt{p}}(y)
\left( Q^{*(y+2)}(x) - 2 Q^{*(y+1)}(x) + Q^{*y}(x) \right), 
\end{equation}
where $\vc{\wt{p}} = (p_3, \ldots, p_n)$.  
\end{lemma}
\begin{proof} Note that the sum of the entries of $\vc{p}_t$ is
constant as $t$ varies, and that $\vc{p_t} = \vc{p}$
for $t = (p_1 - p_2)/2$, while
$\vc{p_t} = ( (p_1 + p_2)/2, (p_1 + p_2)/2, p_3, \ldots, p_n)$
for $t=0$. Writing
$k = p_1 + p_2$, $\bp{p_t}$ can be expressed,
\begin{eqnarray*}
\bp{p_t}(y) & = & 
\left(  \frac{k^2}{4} - t^2 \right) 
\bp{\wt{p}}(y-2)   
+ \left( k \left( 1 - \frac{k}{2} \right)
+2 t^2 \right) \bp{\wt{p}}(y-1) \\
& & + \left( \left( 1 - \frac{k}{2} \right)^2 - t^2 \right) 
\bp{\wt{p}}(y),
\end{eqnarray*}
and its derivative with respect to $t$ is,
$$ \frac{ \partial}{\partial t} \bp{p_t}(y) 
= - 2t \left(  \bp{\wt{p}}(y-2) - 2 \bp{\wt{p}}(y-1) + \bp{\wt{p}}(y) \right).$$
The expression (\ref{eq:master}) for
$\cp{p_t}$ shows that it is
a finite linear combination of compound-sum
probabilities $\bp{p_t}(x)$,
so we can differentiate inside the sum to obtain, 
\begin{eqnarray*}
\frac{ \partial}{\partial t} \cp{p_t}(x)
&  = & \sum_{y=0}^n \frac{ \partial}{\partial t} \bp{p_t}(y)
Q^{*y}(x) \\
& = & - 2t \sum_{y=0}^n \left(  \bp{\wt{p}}(y-2) - 2 \bp{\wt{p}}(y-1) + \bp{\wt{p}}(y) \right) Q^{*y}(x) \\
& = & -2 t \sum_{y=0}^{n-2} \bp{\wt{p}}(y) \left( 
Q^{*(y+2)}(x) - 2 Q^{*(y+1)}(x) + Q^{*y}(x) \right),  
\end{eqnarray*}
since $\bp{\wt{p}}(y) = 0$ for $y \leq -1$ and $y \geq n-1$. \end{proof}

Next we state and prove the equivalent 
of Proposition~\ref{prop:deriv} above:
\begin{proposition} \label{prop:deriv2}
Suppose that the distribution $Q$ on $\Nat$ 
and the compound binomial distribution 
$\mbox{\em CBin}(n,\lambda/n,Q)$
are both log-concave; let 
$\vc{p}=(p_1,p_2,\ldots,p_n)$ be a 
given parameter vector with $n\geq 2$,
$p_1 +p_2+ \ldots + p_n = \lambda>0$,
and $p_1\geq p_2$;
let $W_t$ be a 
random variable with distribution $\cp{p_t}$;
and define, for all $t\in[0,(p_1-p_2)/2],$
the function,
$$E(t):=E[-\log \cp{\pbar}(W_t)],$$
where $\pbar$ denotes the parameter
vector with all entries equal to $\lambda/n$.
If $Q$ satisfies either of the conditions:
$(a)$~$Q$ finite support; or 
$(b)$~$Q$ has tails heavy enough so that,
for some $\rho,\beta>0$ and $N_0\geq 1$, 
we have, $Q(x)\geq \rho^{x^\beta}$,
for all $x\geq N_0$, then
$E(t)$ is continuous for all 
$t\in[0,(p_1-p_2)/2]$,
it is differentiable for 
$t\in(0,(p_1-p_2)/2)$,
and, moreover, $E'(t)\leq 0$ for
$t\in(0,(p_1-p_2)/2)$.
In particular, $E(0)\geq E((p_1-p_2)/2)$.
\end{proposition}

\begin{proof} 
The compound distribution $C_Q\bp{p_t}$ is
defined by the finite sum,
$$
C_Q\bp{p_t}(x)=\sum_{y=0}^n\bp{p_t}(y)Q^{*y}(x),$$
and is, therefore, continuous in $t$. First,
assume that $Q$ has finite support.
Then so does $C_Q\bp{p}$ for any parameter
vector $\vc{p}$, and the continuity and 
differentiability of $E(t)$ are trivial.
In particular, the series defining $E(t)$ 
is a finite sum, so we can differentiate 
term-by-term, to obtain,
\begin{eqnarray}
E'(t)
& = & - \sum_{x=0}^{\infty} \frac{\partial}{\partial t} \cp{p_t}(x) 
	\log \cp{\pbar}(x) \nonumber  \\
& = & 2t \sum_{x=0}^{\infty}
\sum_{y=0}^{n-2} \bp{\wt{p}}(y)
\left( Q^{*(y+2)}(x) - 2 Q^{*(y+1)}(x) + Q^{*y}(x) \right) 
\log \cp{\pbar}(x)
 \label{eq:binstep2}  \\
& = & 2t \sum_{y=0}^{n-2} \sum_{z=0}^{\infty} \bp{\wt{p}}(y) Q^{*y}(z) \sum_{v,w} Q(v) Q(w)
\bigg[ \log \cp{\pbar}(z+v+w) - \log \cp{\pbar}(z+v)  \nonumber \\
& & \hspace*{6.5cm}
-  \log \cp{\pbar}(z+w) + \log \cp{\pbar}(z)
\bigg], \label{eq:binstep3}
\end{eqnarray}
where (\ref{eq:binstep2}) follows by Lemma~\ref{lem:partials}. 
By assumption, the distribution $\cp{\pbar}=\mbox{CBin}(n,\lambda/n,Q)$ 
is log-concave,
which implies that,
for all $z,v,w$ such that $z+v+w$ is in the
support of $\mbox{CBin}(n,\lambda/n,Q)$,
\begin{equation*} 
\frac{ \cp{\pbar}(z)}{\cp{\pbar}(z+v)}
\leq \frac{ \cp{\pbar}(z+w)}{\cp{\pbar}(z+v+w)}.
\end{equation*}
Hence the term in square brackets in equation (\ref{eq:binstep3}) 
is negative, and the result follows.

Now, suppose condition $(b)$ holds on the tails of $Q$.
First we note that the moments of $W_t$ are all uniformly
bounded in $t$: Indeed, for any $\gamma>0$,
\begin{equation}
E[W_t^\gamma]=\sum_{x=0}^\infty 
C_Q\bp{p_t}(x)
x^\gamma 
=
\sum_{x=0}^\infty 
\sum_{y=0}^n\bp{p_t}(y) Q^{*y}(x)
x^\gamma
\leq
\sum_{y=0}^n
\sum_{x=0}^\infty 
Q^{*y}(x)
x^\gamma
\leq C_nq_\gamma,
\label{eq:moment}
\end{equation}
where $C_n$ is a constant depending
only on $n$, and $q_\gamma$ is the
$\gamma$th moment of $Q$, which
is of course finite; recall property~$(ii)$
in the beginning of Section~\ref{sec:comppoi}.

For the continuity of $E(t)$, it suffices to show that 
the series,
\begin{eqnarray}
E(t):=E[-\log \cp{\pbar}(W_t)]=
-\sum_{x=0}^\infty C_Q\bp{p_t}(x)\log C_Q\bp{\pbar}(x),
\label{eq:Eseries}
\end{eqnarray}
converges uniformly. The tail assumption on $Q$ implies
that, for all $x\geq N_0$,
$$1 \geq \cp{\pbar}(x) = \sum_{y=0}^n \bp{\pbar}(y) \qst{y}(x)
\geq \lambda(1-\lambda/n)^{n-1} Q(x)
\geq \lambda(1-\lambda/n)^{n-1} \rho^{x^\beta},$$
so that,
\begin{equation}
0\leq -\log \cp{\pbar}(x)\leq Cx^\beta,
\label{eq:logQ}
\end{equation}
for an appropriate constant $C>0$.
Then, for $N\geq N_0$, the tail of the series
(\ref{eq:Eseries}) can be bounded,
$$0\leq -\sum_{x=N}^\infty C_Q\bp{p_t}(x)\log C_Q\bp{\pbar}(x)
\leq C E[ W^\beta_t{\mathbb I}_{\{W_t\geq N\}}]
\leq \frac{C}{N}E[W_t^{\beta+1}]
\leq \frac{C}{N}C_nq_{\beta+1},
$$
where the last inequality follows from (\ref{eq:moment}).
This obviously converges to zero, uniformly
in $t$, therefore $E(t)$ is continuous.

For the differentiability of $E(t)$, 
note that the summands in (\ref{eq:series}) are
continuously differentiable (by Lemma~\ref{lem:partials}),
and that the series of derivatives converges uniformly
in $t$; to see that, for $N\geq N_0$ we apply
Lemma~\ref{lem:partials} together with the bound
(\ref{eq:logQ}) to get,
\begin{eqnarray*}
\lefteqn{
	\left| \sum_{x=N}^{\infty} \frac{\partial}{\partial t} 
	\cp{p_t}(x) \log \cp{\pbar}(x) \right| 
	} \\
& \leq & 
	2 t \sum_{x=N}^{\infty} 
	\sum_{y=0}^n \bp{\wt{p}}(y)
	\left( Q^{*(y+2)}(x) + 2 Q^{*(y+1)}(x) + Q^{*y}(x) \right) 
	Cx^\beta\\
& \leq & 
	2 C t 
	\sum_{y=0}^n 
	\sum_{x=N}^{\infty} 
	\left( Q^{*(y+2)}(x) + 2 Q^{*(y+1)}(x) + Q^{*y}(x) \right) 
	x^\beta,
\end{eqnarray*}
which is again easily seen to converge to zero
uniformly in $t$ as $N\to\infty$, since
$Q$ has finite moments of all orders.
This establishes the differentiability of $E(t)$
and justifies the term-by-term differentiation
of the series (\ref{eq:series}); the rest of 
the proof that $E'(t)\leq 0$ is the same as in case~$(a)$.
\end{proof}

Note that, as with Proposition~\ref{prop:deriv}, 
the above proof only requires that the compound
binomial distribution $\mbox{CBin}(n,\lambda/n,Q)=\cp{\pbar}$ 
satisfies a property weaker than log-concavity, namely 
that the function,
$\log \cp{\pbar}(x) - \sum_v Q(v) \log \cp{\pbar}(x+v),$
be increasing in $x$.

%


\begin{proof}{\bf of Theorem~\ref{thm:mainber}}
Assume, without loss of generality,
that $n\geq 2$. If $p_1 > p_2$, then 
Proposition~\ref{prop:deriv2}
says that, $E((p_1-p_2)/2)\leq E(0)$, that is,
$$ -\sum_{x=0}^{\infty} 
\cp{p}(x) \log \cp{\pbar}(x)
\leq - \sum_{x=0}^{\infty} 
\cp{p^*}(x) \log \cp{\pbar}(x),$$
where $\vc{p^*} = ((p_1 + p_2)/2, (p_1 + p_2)/2, p_3, \ldots p_n)$
and $\vc{\pbar} = (\lambda/n,\ldots,\lambda/n)$.
Since the expression
$\sum_{x=0}^{\infty} \cp{p_t}(x) \log \cp{\pbar}(x)$
is invariant under permutations of the elements of 
the parameter vectors,
we deduce that it is maximized
by $\vc{p_t} = \pbar$. Therefore,
using, as before, the nonnegativity
of the relative entropy,
\begin{eqnarray*}
H(\cp{p})
&\leq&
	H(\cp{p}) + D(\cp{p} \| \cp{\pbar})\\
& = &
	-\sum_{x=0}^{\infty} 
	\cp{p}(x) \log \cp{\pbar}(x) \\
&\leq&
	-\sum_{x=0}^{\infty}
	\cp{\pbar}(x) \log \cp{\pbar}(x)\\
& = & 
 H( \cp{\pbar} )
\;=\; H(\mbox{CBin}(n,\lambda/n,Q)),
\end{eqnarray*}
as claimed.
\end{proof}

\newpage

\section{Conditions for Log-Concavity} \label{sec:lccond}

Theorems~\ref{thm:mainpoi} and~\ref{thm:mainber} state that 
log-concavity is a sufficient condition for 
compound binomial and compound Poisson distributions 
to have maximal entropy within a natural class. 
Here we give examples of when
log-concavity holds; if the results in this section 
can be strengthened (in particular, 
if Conjecture~\ref{conj:lcconj} can be proved), 
then the class of maximum entropy distributions 
will be accordingly widened.

Below we show that a compound Bernoulli sum  
is log-concave if the parameters are sufficiently 
large, and that compound Bernoulli sums and compound 
Poisson distributions are log-concave if $Q$ is either
supported only on the set $\{1,2 \}$ or is geometric. 


\begin{lemma} \label{lem:lc}
Suppose $Q$ is a log-concave distribution on $\Nat$.

{\em (i)}
The compound Bernoulli distribution
$\cbern{p}{Q}$ is log-concave if and only if
$p \geq \frac{1}{1 + Q(1)^2/Q(2)}$.

{\em (ii)} The compound Bernoulli sum distribution
$\cp{p}$ is log-concave as along as
all the elements $p_i$
of the parameter vector $\vc{p}=(p_1,p_2,\ldots,p_n)$
satisfy $p_i \geq \frac{1}{1 + Q(1)^2/Q(2)}$.
\end{lemma}

\begin{proof} Let $Y$ have distribution
$\cbern{p}{Q}$. Since $Q$ is log-concave
itself, the log-concavity of
$\cbern{p}{Q}$ is equivalent to 
the inequality,
$ \Pr(Y=1)^2 \geq \Pr(Y=2) \Pr(Y=0)$, 
which states that,
$ (p Q(1))^2 \geq (1-p) p Q(2)$, 
and this is exactly the assumption
of~(i).

The assertion in~(ii) follows from~(i),
since the sum of independent log-concave
random variables is log-concave; see, e.g., \cite{karlin3}.
\end{proof}

Next we examine conditions under which a compound
Poisson measure is log-concave. Our argument is based,
in part, on the some of the ideas in
Johnson and Goldschmidt
\cite{johnson17}, and also
in Wang and Yeh \cite{wang3},
where transformations that 
preserve log-concavity are studied.

Note that, unlike for the Poisson distribution, 
it is not the case that every compound Poisson 
distribution CPo$(\lambda,Q)$
is log-concave.
Indeed, for any distribution $P$, considering 
the difference,
$C_Q P(1)^2 - C_Q P(0) C_Q P (2)$,
shows that a necessary condition for $C_Q P$ 
to be log-concave is that,
\begin{equation} \label{eq:nec2}
(P(1)^2 - P(0) P(2))/P(0) P(1) \geq Q(2)/Q(1)^2. \end{equation}
Taking $P$ to be the Po$(\lambda)$
distribution, a necessary condition for CPo$(\lambda,Q)$
to be log-concave is that,
\begin{equation} \label{eq:nec}
\lambda \geq \frac{2 Q(2)}{Q(1)^2},
\end{equation} 
while for $P=\bp{p}$, a necessary condition
for the compound Bernoulli sum 
$C_Q\bp{p}$ to be log-concave is,
$$ \sum_i \frac{ p_i}{1-p_i} + \left(\sum_i \frac{p_i^2}{(1-p_i)^2} \right)
\left(\sum_i \frac{p_i}{1-p_i} \right)^{-1} \geq \frac{2 Q(2)}{Q(1)^2},$$
which, by Jensen's inequality, will hold as long as,
$\sum_i p_i \geq 2 Q(2)/Q(1)^2$.

\begin{theorem} \label{thm:q2pt}
Let $Q$ be a distribution supported on the set $\{ 1, 2 \}$. 

{\em (i)} The compound Poisson distribution {\em CPo}$(\lambda,Q)$ 
is log-concave for all $\lambda \geq \frac{2 Q(2)}{Q(1)^2}$.

{\em (ii)} The distribution $C_Q P$ is log-concave
for any ultra log-concave distribution $P$
with support on $\{0,1,\ldots,N\}$
(where $N$ may be infinite),
which satisfies,
$(x+1) P(x+1)/P(x) \geq
2 Q(2)/Q(1)^2$ for all $x=0,1,\ldots,N$.
\end{theorem}

Note that, the second condition in~(ii) is equivalent to requiring 
that $ NP(N)/P(N-1) \geq 2Q(2)/Q(1)^2$ if $N$ is finite, 
or that $\lim_{x \tends} (x+1) P(x+1)/P(x) 
\geq 2 Q(2)/Q(1)^2$ if $N$ is infinite.

\begin{proof}
Writing $R(y) = y! P(y)$, we know that
$ C_Q P(x) = \sum_{y = 0}^x R(y) \left( Q^{*y}(x)/y! \right).$
Hence, the log-concavity of $C_QP(x)$ is equivalent 
to showing that, 
\begin{equation}
\sum_r \frac{ Q^{*r}(2x)}{r!}
\sum_{y+z = r} R(y) R(z) \binom{r}{y} 
\left( \frac{ Q^{*y}(x) Q^{*z}(x)}{Q^{*r}(2x)}
- \frac{ Q^{*y}(x+1) Q^{*z}(x-1)}{Q^{*r}(2x)} \right)\geq 0,
\label{eq:toabel}
\end{equation}
for all $x\geq 2$, 
since the case of $x = 1$ was 
dealt with previously
by equation~(\ref{eq:nec2}).
In particular, for~(i), taking
$P=\mbox{Po}(\lambda)$, 
it suffices to show 
that for all $r$ and $x$, the 
function,
$$ g_{r,x}(k) := 
\sum_{y+z = r} \binom{r}{y} \frac{ Q^{*y}(k) Q^{*z}(2x-k)}{Q^{*r}(2x)} $$
is unimodal as a function of $k$ 
(since $g_{r,x}(k)$ is symmetric about $x$).

In the general case~(ii), writing $Q(2)=p=1-Q(1)$, we have,
$Q^{*y}(x) = \binom{y}{x-y} p^{x-y} (1-p)^{2y-x}$,
so that,
\begin{equation} \label{eq:exact}
 \binom{r}{y} \frac{ Q^{*y}(k) Q^{*z}(2x-k)}{Q^{*r}(2x)} = \binom{2x-r}{k-y} \binom{2r - 2x}{2y-k},
\end{equation}
for any $p$.
Now, following
\cite[Lemma~2.4]{johnson17} and 
\cite[Lemma~2.1]{wang3}, we use summation by parts 
to show that the inner 
sum in (\ref{eq:toabel})
is positive for each $r$ (except for $r=x$ when $x$ is odd), 
by case-splitting according to the parity of $r$.

(a) For $r = 2t$, we rewrite the inner sum 
of equation~(\ref{eq:toabel}) as,
\begin{eqnarray*} 
\lefteqn{ \sum_{s = 0}^t ( R(t+s) R(t-s) - R(t+s+1) R(t-s-1) ) \times  } \\
& & \left( \sum_{y = t-s}^{t+s} \left( \binom{2x-r}{x-y} \binom{2r-2x}{2y-x} -
\binom{2x-r}{x+1-y} \binom{2r-2x}{2y-x-1} \right) \right), \end{eqnarray*}
where the first term in the
above product is positive by the ultra log-concavity of $P$ 
(and hence log-concavity of $R$), and the second term is positive 
by Lemma~\ref{lem:tech2} below.

(b) Similarly, for $x\neq r = 2t+1$, we rewrite the inner 
sum of equation~(\ref{eq:toabel}) as,
\begin{eqnarray*}  \lefteqn{ \sum_{s = 0}^t ( R(t+s+1) R(t-s) - R(t+s+2) R(t-s-1)) \times  } \\
& &  \left( \sum_{y = t-s}^{t+1+s} \left( \binom{2x-r}{x-y} \binom{2r-2x}{2y-x} -
\binom{2x-r}{x+1-y} \binom{2r-2x}{2y-x-1} 
\right) \right),\end{eqnarray*}
where the first term in the product
is positive by the ultra log-concavity of $P$ 
(and hence log-concavity of $R$) and the second term 
is positive by 
Lemma~\ref{lem:tech2} below.

(c) Finally, in the case of $x=r=2t+1$, 
substituting $k = x$ and $k = x+1$
in (\ref{eq:exact}), combining the resulting
expression with
(\ref{eq:toabel}), and noting
that $\binom{2r-2x}{u}$ is 1 if and only if $u=0$
(and is zero, otherwise), we see that the 
inner sum becomes,
$- R(t+1) R(t) \binom{2t+1}{t}$,
and the summands in 
(\ref{eq:toabel}) reduce to,
$$ - \frac{ {p^x} R(t) R(t+1)}{(t+1)! t!}.$$
However, the next term in the outer sum of equation~(\ref{eq:toabel}), 
$r = x+1$, gives
\begin{eqnarray*}
\lefteqn{ \frac{ p^{x-1} (1-p)^2}{2(2t)!} 
\left[ R(t+1)^2 \left( 2 \binom{2t}{t} - \binom{2t}{t+1} \right)
- R(t) R(t+2) \binom{2t}{t} \right] } \\
& \geq & \frac{ p^{x-1} (1-p)^2}{2 (2t)!} 
R(t+1)^2 \left(  \binom{2t}{t} - \binom{2t}{t+1} \right)
= \frac{p^{x-1} (1-p)^2}{ 2(t+1)! t!} R(t+1)^2.
\end{eqnarray*}
Hence, the sum of the first two terms is positive (and hence the whole sum
is positive) if $R(t+1) (1-p)^2/(2p) \geq R(t)$.

If $P$ is Poisson($\lambda$), this simply reduces to equation~(\ref{eq:nec}), 
otherwise we use the fact
that $R(x+1)/R(x)$ is decreasing.
\end{proof}

\begin{lemma} \label{lem:tech2}

{\rm (a)} If $r = 2t$, for any $0 \leq s \leq t$, the sum,
$$ \sum_{y = t-s}^{t+s} \left( \binom{2x-r}{x-y} \binom{2r-2x}{2y-x} -
\binom{2x-r}{x+1-y} \binom{2r-2x}{2y-x-1} \right) 
\geq 0.$$

{\em (b)} If $x\neq r = 2t+1$, for any $0 \leq s \leq t$, the sum,
$$ \sum_{y = t-s}^{t+1+s} \left( \binom{2x-r}{x-y} \binom{2r-2x}{2y-x} -
\binom{2x-r}{x+1-y} \binom{2r-2x}{2y-x-1} 
\right) \geq 0.$$
\end{lemma}
\begin{proof} The proof is in two stages; first we show that 
the sum is positive for $s = t$, then we show that there 
exists some $S$ such that, as $s$ increases,
the increments are positive for $s \leq S$ 
and negative for $s > S$. The result
then follows, as in \cite{johnson17} or \cite{wang3}.

For both (a) and (b), note that for $s =t$, 
equation~(\ref{eq:exact}) implies that 
the sum is the difference between the coefficients of 
$T^x$ and $T^{x+1}$ in $f_{r,x}(T) = (1+T^2)^{2x-r} (1+T)^{2r-2x}$.
Since $f_{r,x}(T)$ has degree $2x$ and has coefficients 
which are symmetric about $T^x$,
it is enough to show that the coefficients 
form a unimodal sequence. Now, $(1+T^2)^{2x-r} (1+T)$ has
coefficients which do form a unimodal sequence. 
Statement $S_1$ of Keilson and Gerber \cite{keilson} states
that any binomial distribution is strongly unimodal, 
which means that it preserves unimodality on
convolution. This means that 
$(1+T^2)^{2x-r} (1+T)^{2r-2x}$ is unimodal if $r-x \geq 1$,
and we need only check the case $r=x$, when $f_{r,x}(T) = (1+T^2)^r$.
Note that if $r = 2t$ is even, the difference between 
the coefficients of $T^{x}$ and
$T^{x+1}$ is $\binom{2t}{t}$, which is positive.

In part (a), the increments are equal to
$\binom{2x-2t}{x-t+s} \binom{4t-2x}{2t-2s-x}$ multiplied
by the expression,
\begin{eqnarray*}
 2 - \frac{ (x-t-s)(2t-2s-x)}{(x+1-t+s)(2t+2s-x+1)}
- \frac{ (x-t+s)(2t+2s-x)}{(x+1-t-s)(2t-2s-x+1)},
\end{eqnarray*}
which is positive for $s$ small and negative for $s$ large, since 
placing the term in brackets over a common denominator, 
the numerator is of the form $(a-bs^2)$.

Similarly, in part (b), the increments equal
$\binom{2x-2t-1}{x-t+s} \binom{4t+2-2x}{2t-2s-x} $ times
the expression,
\begin{eqnarray*}
 2 - \frac{ (x-t-s-1)(2t-2s-x)}{(x+1-t+s)(2t+2s-x+3)}
- \frac{ (x-t+s)(2t+2+2s-x)}{(x-t-s)(2t+1-2s-x)},
\end{eqnarray*}
which is again positive for $s$ small and negative for $s$ large.
\end{proof}

\begin{theorem} \label{thm:qgeom}
Let $Q$ be a geometric distribution
on $\Nat$. Then $C_Q P$ is log-concave
for any distribution $P$ which is log-concave
and satisfies the condition~{\em (\ref{eq:nec2})}.
\end{theorem}
\begin{proof} If $Q$ is geometric with mean $1/\alpha$, then,
$Q^{*y}(x) = \alpha^y (1-\alpha)^{x-y} \binom{x-1}{y-1}$,
which implies that,
$$ C_Q P(x) = \sum_{y=0}^x P(y) \alpha^y (1-\alpha)^{x-y} \binom{x-1}{y-1}.$$
Condition~(\ref{eq:nec2}) ensures 
that $C_Q P(1)^2 - C_Q P(0) C_Q P (2) \geq 0$, so,
taking $z = y-1$, we need only prove that the sequence,
$$ C(x) := C_Q P(x+1)/(1-\alpha)^x =  \sum_{z=0}^x P(z+1) \left(
\frac{\alpha}{1-\alpha} \right)^{z+1}  \binom{x}{z}$$
is log-concave.
However, this follows immediately from
\cite[Theorem~7.3]{karlin3}, which proves 
that if $\{a_i\}$ is a log-concave sequence, 
then so is $\{b_i\}$, defined by
$ b_i = \sum_{j=0}^i \binom{i}{j} a_j.$
\end{proof}

Finally, based on the discussion in the beginning
of this section, the above results, and some
calculations of the quantities,
$\CPi_{\lambda}(x)^2 - \CPi_{\lambda}(x-1) \CPi_{\lambda}(x+1)$ 
for small $x$,
we make the following conjecture:

\begin{conjecture} \label{conj:lcconj}
The compound Poisson measure {\em CPo}$(\lambda,Q)$ is log-concave,
as long as $Q$ is log-concave and 
$\lambda Q(1)^2\geq 2Q(2)$.  
\end{conjecture}

The condition $\lambda Q(1)^2\geq 2Q(2)$
is, of course, necessary; recall the argument
leading to equation~(\ref{eq:nec}) above.

In closing, we list some known results
that are related to this conjecture and 
may be useful in proving (or disproving) it:
\begin{enumerate}
\item Theorem~2.3 of Steutel and van Harn \cite{steutel2} shows that,
if $\{i Q(i)\}$ is a decreasing sequence, then 
CPo$(\lambda,Q)$ is a unimodal distribution
(recall that log-concavity implies unimodality).
Interestingly, the same condition 
provides a dichotomy of results
in compound Poisson approximation
bounds as developed in \cite{barbour5}:
If $\{i Q(i)\}$ is decreasing the bounds are
of the same form and order as in
the simple Poisson case, while if it is not
the bounds are much larger.
\item Theorem~3.2 of Cai and Willmot \cite{cai} 
shows that if $\{Q(i)\}$ is decreasing
then the distribution function of the compound Poisson 
distribution CPo$(\lambda,Q)$ 
is log-concave.
\item A conjecture similar to Conjecture \ref{conj:lcconj} 
is that, for log-concave $Q$, if CPo$(\lambda,Q)$ is log-concave,
then so is CPo$(\mu,Q)$,
for all $\mu \geq \lambda$. Theorem~4.9 of 
Keilson and Sumita \cite{keilson2} proves 
the related result that, if
$Q$ is log-concave, then, for any $n$,
the ratio,
$$ \frac{ \CPi_{\lambda}(n)}{\CPi_{\lambda}(n+1)} \;\;
\mbox{ is decreasing in $\lambda$.}
$$
\end{enumerate}

\newpage

\section*{Acknowledgement}
We wish to thank Z.\ Chi for sharing
his (unpublished) compound binomial 
counter-example mentioned in 
equation~(\ref{eq:chi})
in the introduction.

\appendix

\section*{Appendix}

\begin{proof}{\bf of Lemma~\ref{lem:moments}}
Recall that, as stated in properties $(ii)$ and~$(iii)$ 
in the beginning of Section~\ref{sec:comppoi},
$Q$ has finite moments of all orders, and 
that the $n$th falling factorial moment
of any ultra log-concave random variable $Y$ 
with distribution $R$ on ${\mathbb Z}_+$ is
bounded above by $(E(Y))^n$. Now for an arbitrary
ultra log-concave distribution $R$, define
random variables $Y\sim R$ and $Z\sim C_Q R$.
If $r_1,r_2,r_3$ denote the first three moments
of $Y\sim R$, then,
\begin{eqnarray} 
E(Z^3)
&=&
	q_3r_1 + 3q_1q_2 E[(Y)_2] + q_1^3 E[(Y)_3]
	\nonumber\\
&\leq&
	q_3r_1 + 3q_1q_2r_1^2 + q_1^3r_1^3.
	\label{eq:3rdmomid}
\end{eqnarray}
Since the map $U_\alpha$ preserves ultra log-concavity
\cite{johnson21}, if $P$ is ultra log-concave then 
so is $R = U_{\alpha} P$, so that (\ref{eq:3rdmomid}) 
gives the required bound for the third moment of
$W_\alpha$, upon noting that the mean of
the distribution $U_\alpha P$ is equal to $\lambda$.

Similarly, size-biasing preserves ultra log-concavity; 
that is, if $R$ is ultra log-concave, then so is $R^{\#}$, since 
$R^{\#}(x+1)(x+1)/R^{\#}(x) = (R(x+2) (x+2)(x+1))/(R(x+1) (x+1))
= R(x+2) (x+2)/R(x+1)$ is also decreasing.
Hence, $R'=(U_\alpha P)^{\#}$ is ultra log-concave, 
and (\ref{eq:3rdmomid}) applies in this case as well.
In particular, noting that the mean of 
$Y'\sim R'= (U_\alpha P)^{\#}=R^{\#}$
can be bounded in terms of the mean of $Y\sim R$ as,
$$E(Y')=\sum_x x\frac{(x+1)U_\alpha P(x+1)}{\lambda}
=\frac{E[(Y)_2]}{E(Y)}\leq\frac{\lambda^2}{\lambda}=\lambda,$$
the bound (\ref{eq:3rdmomid}) yields
the required bound for the third 
moment of $V_\alpha$.
\end{proof}

\end{document}